\documentclass{elsarticle}

\usepackage{hyperref}
\usepackage{bm,algorithm,algorithmic,float,amsmath,amssymb}
\usepackage{graphicx}
\usepackage{tikz}
\usetikzlibrary{shapes.geometric, arrows}
\usepackage{circuitikz}
\usepackage{rotating}
\usetikzlibrary{calc}
\usepackage{color,soul}
\usetikzlibrary{shapes,arrows}
\usepackage{amsthm}
\usepackage{framed}

\newtheorem{proposition}{Proposition}

\newtheorem{corollary}{Corollary}

\usepackage{color}
\usepackage{comment}
\usepackage{booktabs}
\usepackage{multirow}

\usepackage[free-standing-units,use-xspace]{siunitx}

\usepackage{graphicx,epstopdf}
\usepackage{epsfig}
\usepackage{amsfonts}
\floatname{algorithm}{Algorithm}

\begin{document}
	
\begin{frontmatter}
\title{A stochastic approach to estimate distribution grid state with confidence regions}

\author[1]{Rasmus L. Olsen\corref{cor1}\fnref{fn1}}
\ead{rlo@es.aau.dk}
\author[1]{Sina Hassani\fnref{fn2}}
\ead{sinah@es.aau.dk}
\author[1]{Troels Pedersen\fnref{fn3}}
\ead{troels@es.aau.dk}
\author[2]{Jakob Gulddahl Rasmussen\fnref{fn4}}
\ead{jgr@math.aau.dk}
\author[1,3]{Hans-Peter Schwefel\fnref{fn5}}
\ead{schwefel@griddata.eu}

\cortext[cor1]{Corresponding author}
\affiliation[1]{organization={Department of Electronic systems, Aalborg University}, addressline={Fredrik Bajersvej 7A}, city={Aalborg}, postcode={9220}, country={Denmark}}
\affiliation[2]{organization={Department of Mathematical Sciences, Aalborg University}, addressline={Skjernvej 4A}, city={Aalborg}, postcode={9220}, country={Denmark}}
\affiliation[3]{organization={GridData GmbH}, city={Anger},  country={Germany}}


\begin{abstract}
Widely available measurement equipment in electrical distribution grids, such as power-quality measurement devices, substation meters, or customer smart meters do not provide phasor measurements due to the lack of high resolution time synchronisation. Instead such measurement devices allow to obtain magnitudes of voltages and currents and the local phase angle between those. In addition, these measurements are subject to measurement errors of up to few percent of the measurand. In order to utilize such measurements for grid monitoring, this paper presents and assesses a stochastic grid calculation approach that allows to derive confidence regions for the resulting current and voltage phasors. Two different metering models are introduced: a PMU model, which is used to validate theoretical properties of the estimator, and an Electric Meter model for which a Gaussian approximation is introduced. The estimator results are compared for the two meter models and case study results for a real Danish distribution grid are presented.
\end{abstract}

\begin{keyword}
  Grid, state estimation, modelling
\end{keyword}

\end{frontmatter}

\section{Introduction}
Due to the penetration of renewable energy sources into the distribution grid and due to additional loads and production of energy, e.g. from electrification of transport and building heatings, the reliability and efficiency of electricity distribution grids is challenged, and planning and operations of these grids can no longer be performed based on assumptions on loads and generation units. Therefore, grid management, grid optimization, and grid planning using knowledge of the current and historic operational state of the grid is increasingly important \cite{karthik2021, DSOGriplanningpaper}. 
Cable loads and voltage levels that exceed specific limits, must be detected and addressed by proper means, starting by enabling monitoring of the complete distribution grid. 
Furthermore in order to optimize voltage quality and grid efficiency, communication-based hierarchical control systems have been proposed and shown to be efficient \cite{yonghao}. The aforementioned functions of the distribution grid show best performance when the electrical state of the grid is known in all grid locations - while it is infeasible and costly to measure all these grid locations.  The necessary data which shows the condition of the grid and is used for control purposes are called the state variables of the system or in short, states of the system.\par

The large scale of the distribution grid and the presence of communication failure restrict the presence of necessary data for the algorithms. In addition, the impact of the uncertainty of the measured data needs to be considered. Therefore, a process in order to estimate the necessary grid state from the measured data and quantify the impact of measurement errors and missing measurements is necessary. This process is called state estimation in the literature \cite{1,3}. The measured or forecast data used for the estimation process are called the inputs of the system.\par
The state estimation process has for decades been widely used for the high voltage transmission system. These state estimation methods are now starting to become relevant in distribution system state estimation (DSSE), i.e. at medium and low-voltage level. While there are many similarities, the are also some major challenges \cite{review5} when applying DSSE. The key difference is within the observation possibilities, i.e. type and accuracy of measurements, where in the high voltage level, so-called PMUs are able to provide highly accurate synchronized phasor measurements for a large number of measurement location. In contrast to that, measurement devices in the distribution grid entail power quality measurement devices on primary substations, electrical measurement devices on secondary substations, smart meters and inverters at customer connections; none of the latter use accurate clock syncronisation to millisecond level and therefore they only provide magnitudes and local phase angles of RMS voltages and currents. Furthermore, the measurement errors caused by the device type and frequently also from the use of current transformers leads to a higher degree of uncertainty in the measurements; which must be accounted for in the DSSE applied in the MV or LV grid. This uncertainty has strong implications on applications and use of data driven decisions, e.g. planning of new cables and loss reduction \cite{DSOGriLosses, ImadACM}, since not knowing the accuracy of the data may easily lead to wrong decisions taken by the DSO, e.g. replacing the wrong cable or not detecting inefficient grids or faulty grids. Adding confidence regions to data will aid decision makers to know how much trust to put into application results and act accordingly.  

High unobservability of the distribution system is the second challenge of the DSSE, i.e., often the amount of measured data is not enough to estimate the states of the system \cite{observability}. This problem is due to the huge number of connection points of the distribution system compared to the transmission system. The ratio of imaginary and real parts of the line impedance ($x/r$) is low in the distribution grid, making the linear state estimation techniques (DC state estimation) impractical \cite{ratio}. 
Finally, the use of heterogeneous communication networks such as celullar, power-line communication, and low-bandwith technologies such as LoRa causes bandwidth limitations and latency and losses in collection of measurement data. 

According to the literature most of the DSSE methods rely on the weighted least square (WLS) approach \cite{review1,review3,review4}. The WLS method is used in different situations which have different types of states and measurements. It should be stated, that besides the measurement data, the algorithms use forecast data to avoid unobservability and achieve better performance. Other algorithms are used for the DSSE such as Kalman filter, interval state estimation, gradient-based method, and semi-definite programming \cite{6,12,13,14}. 

A wide variety of the states of the DSSE algorithm are used in the literature. Many papers consider the voltage phasors of the nodes as the state variables of the estimation technique. These states are used in both polar and rectangular forms. The polar form of the voltage phasors is used in \cite{1,3,5,6}, while the rectangular form of the voltage phasors is used in \cite{2,4}. Another set of state variables in the DSSE methods are current phasors of the branches. These methods are the most popular methods in the literature \cite{review1}. The current phasors are used in polar \cite{9} and rectangular forms \cite{7,8,10,11,12}. It is worth mentioning that in \cite{10,11}, the slack bus voltage magnitude is included in the states, and in \cite{12}, the slack bus voltage phasor in rectangular form is included in the states.

The inputs of the DSSE algorithms vary widely according to the method and the availability of the data. According to the literature the data are different combinations of active and reactive power of the loads and generators in the distribution grid, node voltage magnitudes, branch current magnitudes, and load current magnitudes \cite{review1,review2,review3,review4,review5}. The availability of the phasor measurement units (PMUs) makes it in principle possible to have the phasor measurements, while their deployment in distribution grids in practice is non-existent due to the tight requirements to time synchronization. Nevertheless, the phasor measurements are considered as the inputs in some DSSE algorithms \cite{10,11}.

On the other hand, the increasing deployment of Digital Electrical Measurement devices in transformer stations, and of Smart Meters and Smart Inverters at customer connections provides an increasing set of measurement locations in distribution grids, while these devices cannot determine absolute phase angles, hence are limited to determine magnitudes of voltages and currents and relative phase angles between the voltage and the current at the measurement point. The general poor clock synchronization of smart meters has previously made SE and DSSE at the low voltage level challenging. The impact of clock offsets on loss calculation applications has been analyzed in \cite{Antonius23}: the authors show that the choice of the specific algorithm for loss calculation has a drastic impact on the sensitivity of the calculated grid losses due to clock deviation errors. The paper furthermore addresses errors in the loss calculation caused by duration of measurement interval and by measurement errors; however, the considered loss calculation algorithms assume that all required input variables have been measured, so no DSSE approach is used in that paper.

In this paper, a DSSE algorithm is proposed based on the available measurements of the distribution systems. The DSSE algorithm is presented to fulfill the main objectives of the state estimation process, i.e.\ the estimation and calculation of confidence regions of the unknown parameters of the system. The DSSE algorithm estimates the node voltage phasors and branch current phasors of the distribution grid. The inputs of the system are the available measured node voltages and the branch currents at a subset of grid locations and these measurements are subject to measurement errors. The setup of the framework enables us to assess different metering models, ranging from PMUs to measurement devices with low synchronisation precision. 


Our contribution in this paper is a methodology to assess confidence regions for model based estimates of voltages and currents in electrical grids. We present a stochastic model framework for carrying out the estimation and derive confidence regions analytically, and subsequently show the use in a real world scenario from a Danish DSO. 

Section \ref{sect:model} describes the system model and the model of the measurement devices. Section \ref{sect:estimator} derives the maximum likelihood estimator and the confidence ellipses for the voltage and current phasors. Section \ref{sect:assessment} introduces the assessment approach and Section \ref{sec:casestudy} applies that approach to a real-life scenario of a distribution grid based on actual measurements.
Finally, Section \ref{sect:conclusions} summarizes the results and presents an outlook to future work.







 
 
 

\section{Stochastic models of errors in realistic electrical measurement scenarios}
\label{sect:model}

\subsection{System Model} \label{sec:model}
Consider a graph  with nodes and edges representing possible measurement points and power lines, respectively. Here, either an MV or an LV grid is considered.  In the MV  case,  nodes represent all MV busbars and sleeves including the  MV side busbar of the HV:MV transformer. For LV grids, nodes represent this first busbar is the LV side of the secondary transformer, the junction boxes (JBs), sleeves and customer connection boxes (CCBs).

%


The grid state  is given by a  voltage phasor for each node and a current phasor for each edge of the graph. This amounts to a number $n$ of voltage phasors $x_1\dots,x_n$ and a number $m$ of current phasors  $x_{n+1}\dots x_{n+m}$.  The $k$\textsuperscript{th} phasor is  of the form
\begin{equation}
    x_k = 
    \begin{cases}
     u_k e^{j\theta_k},& k=1,\dots,n,\text{(voltage)} \\ 
    i_k e^{j(\theta_k+\phi_k)},& k=n+1,\dots,n+m,\text{(current),}
        \end{cases}
\end{equation}
where  $u_k$, $i_k$, $\theta_k $ and $\phi_k$ denote respectively the "true" values of voltage magnitude, current magnitude, voltage phase and local phase angle between voltage and current. Collectively, these $N=n+m$ phasors constitute the state vector $x\in \mathbb C^N$.\footnote{Including both current and voltage phasors of the grid as the state vector facilitates inclusion of noise for all measurements. Furthermore, it enables a generic representation of the equation systems for a variable number of observed voltages or currents. 
}  
For simplicity, we focus on a single-phase representation of the grid.
We remark that generalizations to three phase unbalanced grids are possible; the number of phasors in the state vector will triple in that case. Note however, that another type of measurement error would need to be addressed for the practical application to three phase-grids, namely the wrong phase assignment at measurement devices. The latter is out of scope for this paper.

The state vector fulfills a number of linear constraints given by the grid. Kirchhoff's equations constrain the current phasors and further linear equations link the voltage phasors at two ends of a line with the current phasor on that line and the line impedance. These linear constraints can be summarized in a single matrix $C$ by the equation 
\begin{equation}
Cx = 0,
\label{eq:constraint}
\end{equation}
see \cite{HPSsgcomm2019}\footnote{In busses with injected power, the injected current is summed up with the load current; therefore, Kirchoff's equation holds in this situation. In addition, we aim to estimate all the currents and voltages hence, zero-injected busses are not Kron-reduced in this work.}. For generalizations of the contributing equations to 3-phase grids while remaining  in the space covered with linear equations, see \cite{florinpaper} for instance. 

In most situations it is impossible or impractical to measure the whole state vector directly. In fact,  only noisy measurements for functions of a subset of $x$ can be achieved.   For MV grids, those typically include measurements at the primary and secondary transformers and on a subset of MV lines. For LV grids measurements  typically include the secondary substation, a subset of the CCBs, or possibly in special cases also intermediate JBs. In this paper, we consider measurement devices that collect RMS measurements on voltage magnitude, current magnitude and local phase angle between voltage and current. We refer to these measurement devices as meters for brevity.  Measurement data is collected by wired or wireless communication. Furthermore, the data is  time-stamped and sampled in a synchronized manner so we focus on one particular time-interval across the grid, where clock synchronization errors are neglected here (see \cite{Imad21} for impact of inaccurate clocks). 

We assume that measurements can be obtained as in total $K$ real-valued entities from meters that are placed in the grid. 
The concatenation of all measurements from all meters can then be represented by an $K$-dimensional real vector $y = (y_1,\dots,y_K)$. The measurement vector $y$ is modeled as a stochastic vector depending on a subset of the state variables, expressed as  
\begin{equation}
    y = f(Dx,\epsilon) \in \mathbb R^K,
\end{equation}
with $\epsilon$ being a noise vector and the measurement matrix $D \in \{0,1\}^{K\times N}$ is defined so that the complex vector $Dx$ contains only the phasors for which measurements are obtained.

The estimation problem at hand is to recover the full state vector $x$  from vector  $y$, which contains the accumulated measurements. At its disposal, the estimator has knowledge of the grid topology specified by the matrix $C$ (the linear grid equations mapped from grid structure and grid impedance), the measurement matrix $D$ along with information on the metering function $f$ and  distribution of measurement errors.

\begin{figure}
	\centering
\resizebox{\linewidth}{!}{
\tikzstyle{block} = [draw, fill=white, rectangle,     minimum height=3em, minimum width=5em]
\tikzstyle{sum} = [draw, fill=black, circle, node distance=1cm]
\tikzstyle{input} = [coordinate]
\tikzstyle{output} = [coordinate]
\tikzstyle{pinstyle} = [pin edge={to-,thick,black}]

 \begin{tikzpicture}[auto, node distance=4cm,thick]
    \node [input, name=input] {};
    \node [block, right of=input,  node distance=2cm, pin={[pinstyle]above:$C\in\mathbb R^{Q\times N}$},] (grid) {Grid};
    
    \node [block, right of=grid, node distance=3.3cm, pin={[pinstyle]above:{$D\in\{0,1\}^{K\times N}$}}] (meters) {Meters};
    \node [block, right of = meters,, node distance=3cm] (prep) {Data Prep.};
    \node [block, right of = prep, node distance=3cm] (estimator) {Estimator};
    \node [output, right of=estimator,  node distance=3cm] (output){};
   \node [input, above of=meters,node distance=2cm](d) {} ;
   
    \draw [draw,->] (estimator) -- node[align=center, above]{$\hat x \in \mathbb C^{N}$,\\ conf. regions} (output);
    \draw [->] (grid) -- node [name=x] { $x\in \mathbb C^{N}$}(meters);
    \draw [->] (meters) -- node [name=Dx] { $y$}(prep);
    \draw [->] (prep) -- node [name=Dx] {$z$}(estimator);
\end{tikzpicture}}
	\caption{The power grid is specified by the matrix $C$ and subject to a load condition gives a state vector $x$ of voltages and current phasors. Measured characteristics of voltages and currents yielding a noisy vector $y$ used by the estimator to recover the full state vector $x$ along with confidence regions.}
	\label{fig:attacker_model}
\end{figure}
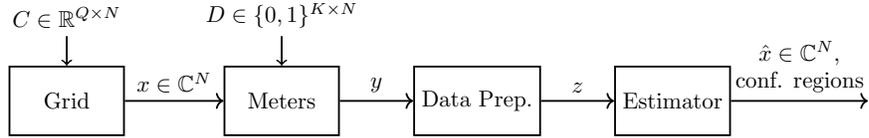

\subsection{Meter Models: PMU and Electrical Meter} \label{sec:EM}
In order to exemplify the measurement model description from the previous section and in preparation for the use-case scenario later  in the paper, we now introduce two different meter models for comparison: (1) A {\bf PMU} measures the two phasors of voltage and current at on specific end of a line; therefore it determines four real-valued measurands, namely the cartesian coordinates of the voltage and current phasor at this measurement location. As an error model, we assume independent complex 2-dimensional Gaussian noise to the Cartesian phasor representations of voltage and current respectively. This model is relevant to validate the theorems presented in the next section, and is used for validation of the confidence regions in the use-case section.

(2) An {\bf Electrical Meter} (EM) measures voltage and current magnitudes and furthermore the local phase angle between voltage and current phasor. Consequently it determines three measurands per measurement location. As measurement error, we here assume independent Gaussian noise on the magnitudes of voltage and currents and on the angle, $\phi$, between voltages and currents, reflecting real world low-cost meter measurement methodologies. 

As the metering models provide different measurands as output (collected in the real-valued vector $y$), a data preparation stage is used to translate the vector $y$ into a unified input $z$ to the estimator, see Figure \ref{fig:datagenerators}. This data preparation stage is specific for each meter model and introduced in the next subsection.

\subsection{Data Preparation for Estimation}\label{sec.modmeas}
The target is to derive (approximate) phasor representations for the current and voltage phasors as approximations for the true phasors in the complex vector $Dx \in \mathbb C^N$ from the measurements in the vector $y\in \mathbb R^M$.  To this end,  we define a vector function $g = (g_1,\dots,g_K)$ where each entry gives the complex-valued form of a voltage or current phasor derived from the measurements

For the PMU measurement model, the function $g$ is trivial to set up as the meter model already delivers the rectangular coordinates of the desired phasors. Therefore, assembling the complex phasors from the two coordinates is sufficient. 

For the EM measuring model, the voltage and current phasors need to be approximated from the measured magnitudes and local phase angle. , i.e.\
\begin{equation}
z_k = g_k( y_k) =
\begin{cases}
    (u_k + \epsilon_{u,k}) e^{j(\theta_k + \epsilon_{\theta,k})}& \text{(voltage)}\\
    (i_k + \epsilon_{i,k})e^{j(\theta_k + \phi_k + \epsilon_{\phi,k} + \epsilon_{\theta,k})}& \text{(current)}
\end{cases}  
\label{eq:5}
\end{equation}
Here,  $\epsilon_{u,k},\epsilon_{i,k},\epsilon_{\theta,k}$ and $\epsilon_{\phi,k}$ denote the respective measurement errors on magnitudes or phase angles (which are already added in the metering model itself). Measurement errors at meter $k$ are regarded as independent random variables with zero mean (unbiased errors). Furthermore, measurement errors from different meters are considered independent.

Since the EM metering model does not determine the voltage angle, extra assumptions are needed to define $z_k$. In particular,  in this case we introduce a \emph{pseudo-measurement} by assuming the voltage phase to be zero. This is a good approximation since the voltage phase $\theta_k$ is usually small in single-phase representations of distribution grids, when the reference point (e.g. the substation) is set to phase-angle zero.  Thereby, we obtain a pseudo measurement of the form \eqref{eq:5} with $\theta_k + \epsilon_{\theta,k}= 0$. 
Finally, we obtain a vector $z = (z_1,\dots,z_K)$ derived from measurements and influenced by the introduction of (pseudo-) measurements (here for the voltage phase angle).

\subsection{Complex Gaussian Approximation of measurement errors provided by EM model} \label{sect:covar}
The probability distribution of $z$ is unknown in general.  For the EM model, even if the distribution of the different $\epsilon_{.k}$ are known, the distribution of $z_k$ is unavailable due to the non-linear transform $g_k$ in \eqref{eq:5}.  
 For mathematical tractability, however, we \emph{approximate} the distribution of $z_k$ by a complex normal distribution defined such that it gives the same first and second moments as $z_k$, i.e\ 
\begin{equation}
    z_k \sim CN(\mu_k,\Sigma_{1,k},\Sigma_{2,k})
\end{equation}
with mean $\mu_k$, variance $\Sigma_{1,k} = \mathrm{Cov}(z_k,z_k)$ and pseudo variance $\Sigma_{2,k} = \mathrm{Cov}(z_k,z_k^*)$, with $^*$ being the complex conjugate.
The assumption of independent errors allows us to state the approximation for the distribution of  the vector $z$ as 
\begin{equation}
z \sim CN(\mu,\Sigma_1,\Sigma_2) 
\end{equation}
where the mean vector $\mu$ contains $\mu_1,\dots, \mu_K$ and the covariance matrix $\Sigma_1$ is a  diagonal matrix with entries $\Sigma_{1, k},~k = 1,\dots,K$. The pseudo-covariance matrix $\Sigma_2$  is also diagonal and defined similarly.


The mean, covariance and pseudo-covariance for $z_k$ can be computed for a given distribution of measurement errors by use of Proposition~\ref{propMeanCovariacePseudocovariance} given in appendix. We obtain  the  mean as
 \begin{equation*}
     \mu_k = \mathbb E[z_k] =
     \begin{cases}
     u_k e^{j\theta_k} c_{\epsilon_{\theta,k}}(1)& \text{(voltage)}\\
     i_k e^{j(\theta_k+\phi_k)}c_{\epsilon_{\theta,k}}(1)c_{\epsilon_{\phi,k}}(1)& \text{(current)},
          \end{cases}
 \end{equation*}
 the variances as
\begin{align}\label{eq:sigma1}
    \Sigma_{1,k} 
  &=
      \begin{cases}
     (1-|c_{\epsilon_{\theta,k}}(1)^2|)u_k^2 + \sigma_{u,k}^2 & \text{(voltage)}\\[1ex]
     (1-|c_{\epsilon_{\theta,k}}(1)^2c_{\epsilon_{\phi,k}}(1)^2|)i_k^2 + \sigma_{i,k}^2& \text{(current)}
          \end{cases}
    \end{align}
and pseudo-variances as
\begin{align}\label{eq:sigma2}
    \Sigma_{2,k}\! 
    &=\! 
\begin{cases}
    e^{2j\theta_k}\left((u_k^2+\sigma_{u_k}^2)c_{\epsilon_{\theta,k}}(2) - u_k^2 c_{\epsilon_{\theta,k}}(1)^2\right)\!\!  & \text{(voltage)}\\[1ex]
    e^{2j(\theta_k+\phi_k)}\Big((i_k^2+\sigma_{i_k}^2)c_{\epsilon_{\theta,k}}(2)c_{\epsilon_{\phi,k}}(2)\\ \qquad - i_k^2 c_{\epsilon_{\theta,k}}(1)^2c_{\epsilon_{\phi,k}}(1)^2\Big) & \text{(current)},
\end{cases}     
\end{align}
 where $c_\epsilon(t) = \mathbb{E}[e^{jt\epsilon}]$ denotes the characteristic function of a random variable $\epsilon$.

The quality of the complex normal distribution approximation depends on the true distribution of $z$. For example, in the EM model, we assume that the measurement error of the measured magnitudes  and of the measured angle are  normally distributed. In such case for large values of the variance of the angle, the true distribution will have a banana shape, and the resulting approximation will be poor, while for small values of the variance of the angle, the approximation is good.

In Figure \ref{fig:datagenerators} we illustrate how the two models, a) the PMU measurement model and b) the EM measurement model, are implemented. In essence, we need to do more than just creating a measurement vector $y$, but also need preparation and calculation of the co-variance matrices that are required for the estimator, which we will describe in the subsequent section. For a real setting, we will receive the measurement vector $y$ and will have to calculate the covariance matrices based on the measurements and based on properties of the measurement device, see Section \ref{sect:assessment}. 

\begin{figure*}
        \centering
         \begin{tikzpicture}
             \node at (-1.6,0) [text width=1cm] (start) {$\sigma_I,\sigma_U$ $\sigma_{\theta},\sigma_{\phi}$};
             \node at (0,0) [rectangle,draw,thick,rounded corners, minimum width=1cm, minimum height=1cm, text centered]  (covar) {CoVar};
             \node at (-1.6,1.2) (x) {$x$};
             \node at (1.3,1.2) [rectangle,draw,thick,rounded corners, minimum width=1cm, minimum height=1cm, text centered]  (dx) {$Dx$};
             \node at (2.2,0) [rectangle,draw,thick,rounded corners, minimum width=1cm, minimum height=1cm, text centered]  (pmu) {PMU};
             \node at (4.5,0) [rectangle,draw,thick,rounded corners, minimum width=1cm, minimum height=1cm, text centered]  (DP) {Data Prep.};
            \draw[thick,->] (start) -- (covar);
            \draw[thick,->] (covar) -- (pmu)  node[draw=none,fill=none,midway,above] {$\Sigma_1,\Sigma_2$};
            \draw[thick,->] (x) -- (dx);
            \draw[thick,->] (dx) -| (pmu);
            \draw[thick,->] (pmu) -- (DP) node[draw=none,fill=none,midway,above] {$y_{PMU}$};
            \draw[thick,->] (DP) -- (6,0) node[draw=none,fill=none,above] {$z$};
            \draw[thick,->] (covar) -- (1.2,0)|- (6,-1) node[draw=none,fill=none,above] {$\Sigma_1,\Sigma_2$};
            \node at (2.35,0.2) [rectangle,draw,thin,dashed,rounded corners, minimum width=6.3cm, minimum height=3.2cm, text centered,label=(a) PMU Data Generator]  (Gen) {};
         \end{tikzpicture}
         \label{fig:pmu}
    \hfill
        \begin{tikzpicture}[node distance=1cm]
            \node at (-0.5,0) [text width=1cm] (start) {$\sigma_I,\sigma_U$ $\sigma_{\theta},\sigma_{\phi}$};
            \node at (-0.5,+1.2) (x) {$x$};
            \node at (3,-1.2)[rectangle,draw,thick,rounded corners, minimum width=1cm, minimum height=1cm, text width=1.5cm, text centered]  (covar) {CoVar \ref{sect:covar}};
             \node at (1.2,+1.2) [rectangle,draw,thick,rounded corners, minimum width=1cm, minimum height=1cm, text width=1cm, text centered]  (dx) {$Dx$ \ref{sec:model}};
             \node at (2,0) [rectangle,draw,thick,rounded corners, minimum width=1cm, minimum height=1cm, text width=1cm, text centered]  (em) {EM \ref{sec:EM}};
             \node at (4.3,0) [rectangle,draw,thick,rounded corners, minimum width=1cm, minimum height=1cm, text width= 1.5cm, text centered]  (DP) {Data Prep. \ref{sec.modmeas}};
            \draw[thick,->] (start) -- (em);
            \draw[thick,->] (start)--(0.8,0) |- (covar);
            \draw[thick,->] (covar) -- (6,-1.2) node[draw=none,fill=none,above] {$\Sigma_1,\Sigma_2$};
            \draw[thick,->] (x) -- (dx);
            \draw[thick,->] (dx) -| (em);
            \draw[thick,->] (em) -- (DP) node[draw=none,fill=none,midway,above] {$y_{EM}$};
            \draw[thick,->] (DP) -- (6,0) node[draw=none,fill=none,above] {$z$};
            \node at (2.8,0) [rectangle,draw,thin,dashed,rounded corners, minimum width=5cm, minimum height=4cm, text centered,label=(b) EM Data Generator]  (Gen) {};
        \end{tikzpicture}
         \label{fig:EM}
\caption{Data generation considering the (a) PMU  and (b) Electrical Meter (EM) models. The PMU blocks adds 2-dimensional Gaussian noise as measurement error to the selection of phasors according to the measurement locations identified by $D$. The parameters for the 2-D Gaussian noise are obtained from standard deviations of the measurement errors for the magnitudes of voltages and current. The same calculation is also done in the EM data generator, but without any impact on internal blocks; so these parameters are only given as output of the EM data generator, as they are later needed in the estimator, see next section. Note that the data preparation stage of the EM model also involves a pseudo-measurement for the voltage angles. }
\label{fig:datagenerators}
\end{figure*}
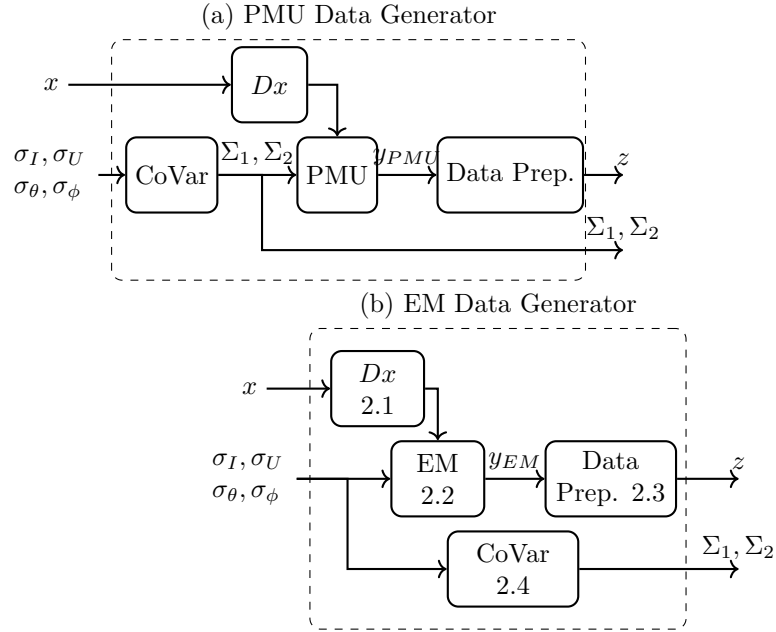

\section{Maximum Likelihood Estimation of Grid State}
\label{sect:estimator}
We now turn to deriving a maximum likelihood estimator for the entire grid state vector $x$ based on the phasor vector $z$ representing derived observations of a smaller number of directly observable state variables $Dx$. The methods for obtaining the maximum likelihood estimator and its properties are an extension of the methods used in \cite{HPSsgcomm2019}, but with a more general model and with additional results. In this section, the measurement model is given by
\begin{equation}\label{eq:modelDx}
z \sim CN(Dx,\Sigma_1,\Sigma_2).
\end{equation}
To do this, it is assumed that the linear constraint \eqref{eq:constraint} along with the noise covariance and pseudo-covariance matrices are available.

For generality,  we derive the estimator under slightly generalized assumptions. Specifically,  we generalize the constraint~\eqref{eq:constraint} to
\begin{equation}\label{eq:constraint_c}
Cx = c,
\end{equation}
and arbitrary complex values in $D$, $c$, $C$, $\Sigma_1$ and $\Sigma_2$, provided that the matrix 
\begin{equation*}
    P = \Sigma_1^*-\Sigma_2^H\Sigma_1^{-1}\Sigma_2
\end{equation*}
is positive semi-definite where superscript $H$ denotes complex conjugate transpose. Furthermore, we denote $W=\Sigma_2 (\Sigma_1^*)^{-1}$ as a convenient short notation, since both $W$ and $P$ appears in various places in the below results. 

For brevity, we formulate the estimator using complex augmented vectors and matrices. Given a general complex vector $v$ or matrices $B_1$ and $B_2$, their augmented versions are marked by an overbar and defined as
\begin{equation*}
    \bar v = \begin{pmatrix}
    v\\v^*
    \end{pmatrix}, \qquad 
    \overline{(B_1,B_2)} = \begin{pmatrix}
    B_1 & B_2\\
    B_2^* & B_1^*
    \end{pmatrix}.
\end{equation*}
In augmented notation the model \eqref{eq:modelDx} is written as
\[
\bar z \sim CN(\bar D \bar x, \bar \Sigma),
\]
where $\bar D = \overline{(D,O)}$, $\bar\Sigma = \overline{(\Sigma_1, \Sigma_2)}$, and $O$ denotes an appropriately sized matrix of zeros.

We estimate $x$ by maximum likelihood estimation, or equivalently least squares estimation, by minimising
\[
(\bar z - \bar D \bar x)^H \bar\Sigma^{-1} (\bar z - \bar D \bar x),
\]
subject to the constraints $c=Cx$. This is solved by minimising the Lagrange function
\begin{equation}\label{eq.lagrange}
\mathcal{L}(x,\lambda) = (\bar z - \bar D \bar x)^H \bar\Sigma^{-1} (\bar z - \bar D \bar x) + \Re(2\lambda^H (Cx-c)).
\end{equation}
This leads to the following result:
\begin{proposition}\label{prop.MLeq}
Assume \eqref{eq:modelDx} and \eqref{eq:constraint_c}. Then the maximum likelihood estimate is a solution to the equation 
\begin{equation}\label{eq.MLeq}
\begin{pmatrix} \bar g \\ \bar c \end{pmatrix}
= 
\begin{pmatrix} \bar G & \bar C^H \\ \bar C & O \end{pmatrix}
\begin{pmatrix} \bar x \\ \bar \lambda \end{pmatrix},
\end{equation}
where $\bar G = \overline{(G_1,G_2)}$, $\bar C = \overline{(C,O)}$,
$G_1 = 2D^H(P^*)^{-1}D$,
$G_2 = -2D^H(P^*)^{-1} W D^*$,
and
$
g = 2 D^H (P^*)^{-1} (z-Wz^*).
$
\end{proposition}

\begin{proof}
Taking the derivative of the Lagrange function \eqref{eq.lagrange} with respect to $x^*$, and simplifying this, we obtain
\begin{equation}\label{eq.difflag}
\frac{\partial \mathcal{L}(x,\lambda)}{\partial x^*} = -g + G_1 x + G_2 x^* + C^H\lambda.
\end{equation}
Equating this to zero and using augmented matrices and vectors, we get
\[
\bar g = \bar G \bar x + \bar C^H\bar\lambda.
\]
Combining this with the constraint equation \eqref{eq:constraint_c} on augmented form, i.e. $\bar c = \bar C \bar x$, we obtain \eqref{eq.MLeq}.
\end{proof}

Throughout this paper we will assume that the matrix 
\begin{equation}\label{eq.A}
A=
\begin{pmatrix} \bar G & \bar C^H \\ \bar C & O \end{pmatrix}
\end{equation}
appearing in \eqref{eq.MLeq} is invertible. If this is not the case, then the solution to \eqref{eq.MLeq} is not unique. Even if it is not unique, the solution may still provided valuable information, e.g.\ some entries in the estimate may still be unique, but we will leave this issue for future study. When $A$ is invertible,
we can formulate the inverse of $A$ as a block matrix with blocks of the same size as $A$, i.e.\
\begin{equation}\label{eq.Ainv}
A^{-1} = \begin{pmatrix} \bar F_{11} & \bar F_{12} \\  \bar F_{21} & \bar F_{22} \end{pmatrix},
\end{equation}
and we can get a closed form solution for the maximum likelihood estimate of $x$. 

\begin{proposition}\label{prop.MLexpr}
Assume \eqref{eq:modelDx} and \eqref{eq:constraint_c}. If $A$ given by \eqref{eq.A} is invertible, the maximum likelihood estimate is given by
\[
\hat{\bar x} = \bar F_{11}\bar g + \bar F_{12} \bar c
\]
\end{proposition}

\begin{proof}
This follows directly from Proposition~\ref{prop.MLeq} by multiplying \eqref{eq.MLeq} by \eqref{eq.Ainv}.
\end{proof}
In principle closed form expressions can be obtained for $\bar{F}_{11}$, $\bar{F}_{12}$, etc.\ by using block inversion on $A$, but this requires $\bar{G}$ to be invertible, and a necessary requirement for this is that $D$ has full column rank. This is rarely fulfilled for the application in the present paper, and instead $A$ can be inverted numerically, and the submatrices can then be extracted from $A^{-1}$. Note that $A$ is a sparse matrix, which makes the numerical inversion faster. When $A$ is invertible, the following proposition gives the distribution of the maximum likelihood estimate. 

\begin{proposition}\label{prop.MLdist}
Assume \eqref{eq:modelDx} and \eqref{eq:constraint_c}. If A is invertible, then
\[
\hat{\bar x} \sim CN(\bar x, 2\bar F_{11}).
\]
\end{proposition}

\begin{proof}
First notice that $\bar g = \bar J \bar z$, where 
\[
\bar J = 
\overline{(2D^H(P^*)^{-1},-2D^H(P^*)^{-1}W)}
\]
and thus from Proposition~\ref{prop.MLexpr}, we get that $\hat{\bar x} = \bar F_{11}\bar J\bar z + \bar F_{12} \bar c$ is a widely linear transformation of $\bar z$. Hence from the distributional assumptions on $z$, i.e. \eqref{eq:modelDx}, we obtain that
\[
\hat{\bar x} \sim CN(\bar F_{11}\bar J \bar D\bar x + \bar F_{12}\bar c, \bar F_{11}\bar J \bar\Sigma \bar J^H \bar F_{11}).
\]
To simplify the expectation, observe that expanding the identity $A^{-1}A=I$ into its blocks, we obtain
\begin{equation}\label{eq.AAI}
\bar F_{11}\bar G + \bar F_{12}\bar C = I \quad \text{and} \quad \bar C \bar F_{11} = O.
\end{equation}
The first equation in \eqref{eq.AAI} together with the observation that $\bar J \bar D = \bar G$ implies that 
\[
\mathbb{E}\hat{\bar x} = F_{11}\bar J \bar D\bar x + \bar F_{12}\bar c = (F_{11}\bar G + \bar F_{12}\bar C)\bar x = \bar x.
\]
To simplify the augmented covariance matrix, firstly observe that combining both equations in \eqref{eq.AAI}, we obtain that
\begin{equation}\label{eq.FGFF}
\bar F_{11} \bar G \bar F_{11} = \bar F_{11}.
\end{equation}
Secondly, expanding and simplifying, we get that
\begin{equation}\label{eq.JSJG}
\bar J \bar \Sigma \bar J^H = 2\bar G.
\end{equation}
From \eqref{eq.FGFF} and \eqref{eq.JSJG}, it immediately follows that the covariance matrix of $\hat{\bar x}$ is given by $2\bar F_{11}$.
\end{proof}

We note that Proposition~\ref{prop.MLdist} implies that the maximum likelihood estimator is unbiased. Furthermore, it shows that the estimator is efficient, i.e. has a minimal covariance matrix, according to the following proposition, which gives the Cramer-Rao lower bound.

\begin{proposition}\label{prop.CRLB}
Assume \eqref{eq:modelDx} and \eqref{eq:constraint_c}. If A is invertible, then
\[
\text{Cov}(\tilde{\bar x}) \succeq 2\bar F_{11},
\]
where $\succeq$ means that $\text{Cov}(\tilde{\bar x}) - 2\bar F_{11}$ is positive semi-definite, and $\tilde{\bar x}$ is an arbitrary unbiased estimator of $\bar x$.
\end{proposition}

\begin{proof}
By Theorem 1 in \cite{jagannathamrao04}, the Cramer-Rao lower bound in the complex valued and constrained case is given by
\[
\text{Cov}(\tilde{\bar x}) \succeq U(U^H I(\bar x) U)^{-1} U^H,
\]
where $I(\bar x)$ is the Fisher information matrix, and $U$ is a $2N\times(2N-2Q)$ matrix fulfilling 
\begin{equation}\label{eq.U}
U^T U = I \quad \text{and} \quad \left(\frac{\partial}{\partial \bar x} h(\bar x)\right) U = O
\end{equation}
and $h(\bar x)=0$ is the constraint. 

Since $h(\bar x)=\bar C\bar x - \bar c$, the latter part of \eqref{eq.U} becomes $\bar CU=O$, and thus the column space of the matrix $U$ is the orthogonal complement of the row space of $\bar C$, since $A$ is invertible and therefore $\text{rank}(\bar C)=2Q$ implying $\text{rank}(U)=2N-2Q$, i.e.\ $U$ has full rank. Now observe that since the $2N\times 2N$ matrix $\bar F_{11}$ fulfills the second equation in \eqref{eq.AAI}, $\bar F_{11}$ has the same column space, provided $\text{rank}(\bar F_{11})=2N-2Q$. We verify the rank: Firstly, by the second equation in \eqref{eq.AAI}, we get that $\text{rank}(\bar F_{11}) \leq 2N-2Q$. Secondly, by the first equation in \eqref{eq.AAI}, we get that
\[
2N \leq  \text{rank}(\bar F_{11} \bar G) + \text{rank}(\bar F_{12}\bar C) 
\leq \text{rank}(\bar F_{11}) + \text{rank}(\bar C)
\]
so $\text{rank}(\bar F_{11}) \geq 2N-2Q$. Since $U$ and $\bar F_{11}$ thus have the same column space, we can use derivations similar to those immediately following (24) in \cite{gormanhero90} to get that 
\begin{equation}\label{eq.CFFFF}
\text{Cov}(\tilde{\bar x}) \succeq \bar F_{11}(\bar F_{11}^H I(\bar x) \bar F_{11})^+ \bar F_{11}^H,
\end{equation}
where we have first exchanged the inverse with a Moore-Penrose pseudoinverse shown as a $^+$. 

By similar derivations as is used in \eqref{eq.difflag} for the first order derivatives of the Lagrange function, the second order derivative of the log likelihood function can be obtained, and from this we can obtain the Fisher information 
\[
I(\bar{x}) = E\left(-\frac{\partial}{\partial \bar x}\frac{\partial}{\partial \bar x^H} \log l\right) = \frac{1}{2}\bar G.
\]
Inserting this into \eqref{eq.CFFFF} and using \eqref{eq.FGFF}, we get that 
\begin{equation*}
\text{Cov}(\tilde{\bar x}) \succeq \bar 2F_{11}(\bar F_{11}^H \bar G \bar F_{11})^+ \bar F_{11}^H = 2F_{11}\bar F_{11}^+ \bar F_{11}^H = 2\bar F_{11}.
\end{equation*}
\end{proof}

We can obtain confidence intervals and ellipses from Proposition~\ref{prop.MLdist} in the following way. Let $F_1$ and $F_2$ denote blocks in $\bar{F}_{11}=\overline{(F_1,F_2)}$.
Then we get by Proposition~\ref{prop.MLdist} that the covariance matrix for $\hat x_i$ (formulated as a two-dimensional real vector instead of a complex number) is given by
\begin{equation}\label{eq.covmatxhati}
\begin{pmatrix}
\Re((F_1)_{ii} + (F_2)_{ii}) &
\Im(-(F_1)_{ii} + (F_2)_{ii}) \\
\Im((F_1)_{ii} + (F_2)_{ii}) &
\Re((F_1)_{ii} - (F_2)_{ii}) 
\end{pmatrix}.
\end{equation}
From this, it follows that a confidence interval with confidence level $1-\alpha$ for the real and imaginary parts of $x_i$ are given by
\begin{align*}
&\Re(\hat x_i) \pm z_{\alpha/2} \sqrt{\Re((F_1)_{ii} + (F_2)_{ii})},\\
&\Im(\hat x_i) \pm z_{\alpha/2} \sqrt{\Re((F_1)_{ii} - (F_2)_{ii})}.
\end{align*}
Furthermore it follows that a confidence ellipse for $x_i$ is given by the ellipse with 
\begin{itemize}
    \item center $\hat x_i$, 
    \item angle $\tan^{-1}(v_2/v_1)$,
    \item major axis $\sqrt{e_1\chi_{1-\alpha,2}^2}$, and    
    \item minor axis $\sqrt{e_2\chi_{1-\alpha,2}^2}$,
\end{itemize}
where $e_1$ is the largest eigenvalue of \eqref{eq.covmatxhati}, $e_2$ is the smallest eigenvalue, and $(v_1,v_2)$ is the eigenvector corresponding to the largest eigenvalue.

\section{Assessment Approach}
\label{sect:assessment}

Figure \ref{fig:assessmentapproach} shows the details of how the estimator is applied and assessed.  We use a real-life distribution grid with smart meter measurements that determine the load scenario in a load flow calculation to determine the 'true' voltage and current phasors everywhere in the grid. This load-flow calculation together with the structural data and the load data is called the Reference Grid, see Figure \ref{fig:assessmentapproach}. The load-flow calculation in the  RGM is based in a standard Newton-Raphson algorithm to solve the power flow equations, see e.g. \cite{saadat1999power}.

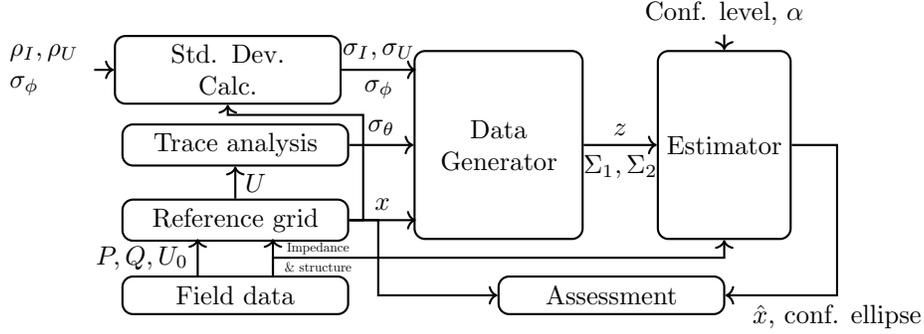
\begin{figure*}
    \centering
    \begin{tikzpicture}
        \node at (5,0) [rectangle,draw,thick,rounded corners, minimum width=1cm, minimum height=2.5cm, text centered]  (es)  {Estimator};
        \node at (5,1.75) (conf) {Conf. level, $\alpha$};
        \draw[thick,->] (conf) -- (es);
        \node at (2,0) [rectangle,draw,thick,rounded corners, minimum width=2cm, minimum height=2.5cm, text centered,text width=2cm]  (mg)  {Data Generator};
        \node at (3.5,-2) [rectangle,draw,thick,rounded corners, minimum width=3cm, minimum height=0.5cm, text centered,text width=2cm]  (ass)  {Assessment};
        \node at (-1.5,-2) [rectangle,draw,thick,rounded corners, minimum width=3cm, minimum height=0.5cm, text centered,text width=2cm]  (fd)  {Field data};
        \node at (-1.5,-1) [rectangle,draw,thick,rounded corners, minimum width=3cm, minimum height=0.5cm, text centered,text width=2.5cm]  (rg)  {Reference grid};
        \node at (-1.5,0) [rectangle,draw,thick,rounded corners, minimum width=3cm, minimum height=0.5cm, text centered,text width=2.5cm]  (ta)  {Trace analysis};
        \node at (-1.6,1) [rectangle,draw,thick,rounded corners, minimum width=3cm, minimum height=0.5cm, text centered,text width=2cm]  (ec)  {Std. Dev. Calc.};
        \draw[thick,->] (rg)-- (ta) node[draw=none,fill=none,midway,right] {$U$};
        \draw[thick,->] (ta)-- (mg) node[draw=none,fill=none,midway,above] {$\sigma_{\theta}$};
        \draw[thick,->] (ec)-- (0.9,1) node[draw=none,fill=none,midway,above] {$\sigma_I,\sigma_U$}node[draw=none,fill=none,midway,below] {$\sigma_{\phi}$};
        \draw[thick,->] (mg)-- (es) node[draw=none,fill=none,midway,below] {$\Sigma_1,\Sigma_2$}node[draw=none,fill=none,midway,above] {$z$};
        \draw[thick,->] (-2,-1.75)-- (-2,-1.25) node[draw=none,fill=none,midway,left] {$P,Q,U_0$};
        \draw[thick,->] (-1,-1.75)-- (-1,-1.25) ;
        \draw[thick,->] (-1,-1.75)-- (-1,-1.5)-|(es) node[draw=none,fill=none,at start,xshift=0.6cm,yshift=-0.02cm,above,scale=0.5] {Impedance } node[draw=none,fill=none,at start,xshift=0.6cm,below,scale=0.5] {\& structure} ;
        \draw[thick,->] (rg)-- (0.9,-1) node[draw=none,fill=none,midway,above] {$x$};
        \draw[thick,->] (rg)-| (0.2,0.4)-|(ec);
        \draw[thick,->] (rg)-- (0.4,-1)|- (ass);
        \draw[thick,->] (es)-- (6.5,0) |- (ass) node[draw=none,fill=none,midway,below] {$\hat{x}$, conf. ellipse};
        \node at (-4,1) [text width=1cm] (start) {$\rho_I,\rho_U$ $\sigma_{\phi}$};
        \draw[thick,->] (start) -- (ec);
    \end{tikzpicture}
    \caption{Overview of the assessment approach: the field measurements of active and reactive power, the measured voltage at the LV side of the substation and the case study grid structure  are used as input to a reference grid model in order to derive the true state vector $x$ of the grid. A subset of the components of $x$ are then used to generate measurements in the data generator according to the PMU or EM model, see Figure \ref{fig:datagenerators}. The output of the Data Generator is then channeled into the estimator and its results are compared to the original state vector $x$.}
\label{fig:assessmentapproach}
\end{figure*}

The true state vector $x$ is then used by the Data Generator to generate measurements according to the PMU or EM model, shown in Figure \ref{fig:datagenerators}. The preprocessed data $z$ obtained from the measurements is then fed into the estimator from Section \ref{sect:estimator}. 

The estimator provides the estimated state vector $\hat{x}$ together with the confidence ellipses. This output is then compared with the true state vector $x$. The actual metrics for comparison will be provided in the next section together with the results.

The data generator requires as input the standard deviations of the voltage magnitude, of the current magnitude, of the phase angle, and of the voltage angle. The standard deviation of the phase angle, $\phi$, between voltage and current phasors has been stated in the literature as $\sigma_\phi=10^{-2}~rad$ \cite{10,11}, so this value is used in the case study later.

 The standard deviation of the voltage angle, $\sigma_\theta$, is in our assessment approach obtained by calculating the grid scenario for a set of realistic load conditions and calculating the empirical standard deviation of the samples of $\theta$ over all grid locations and over all load conditions. By using actual measurement data as realistic load conditions we get $\sigma_{\Theta}= 0.003~rad$ later in the case study. 

The errors on voltage magnitude and current magnitude are specified by relative errors, $\rho_U$ and $\rho_I$. For voltages, the relative error is defined in reference to the nominal voltage, which is $U_{nominal}=400V$ later in the case study. For currents, the relative error is in relation to the true current magnitude at each measurement location, i.e. the calculation of the standard deviation of the measurement error for current magnitudes will depend on the actual current magnitude obtained from the true state vector $x$ in the assessment approach in Figure \ref{fig:assessmentapproach}.  The standard deviations $\sigma$ of error for voltage magnitudes and current magnitudes are chosen in a way, that $\beta=99\%$ of the erroneous samples are within a fraction $\rho$ of the nominal value (for voltages), respectively the true value $\mu$ (for currents), i.e.
\begin{equation}\label{eq:prsd}
Pr(X \in (\mu(1-\rho),\mu(1+\rho)))=\beta.
\end{equation}
For the magnitudes of voltages and currents later in the case study, we use respectively, $\rho_U=1\%$ and $\rho_I=3\%$, \cite{8,review2}. 
Using the $(1-\beta)/2=0.5\%$ quantiles of the normal distribution, we 
get $\sigma_U$ and $\sigma_I$ from: 
$$\sigma_U r_0 = \mu_U\rho_U,\qquad \sigma_I r_0 = \mu_I\rho_I$$
with $r_0=norminv((1+\beta)/2,0,1)$.\par

Note that in the general case, different values of $\rho_U$ and $\rho_I$ can be used for different measurement locations, e.g. reflecting different types of measurement devices, so these relative errors could be vectors. In the case study of the next section, we assume the use of identical measurement devices at the different locations, so the scalars, $\rho_U$ and $\rho_I$ introduced above are used at every measurement location. 
´

\section{Case-Study Introduction}
\label{sect:usecase}
To validate the methods we apply the assessment approach on a case study using a real-life Danish distribution grid \cite{N2DG_D5.3}.

\subsection{Grid Scenario}
 The grid covers a typical small sized Danish town, supplying energy to several households, a school, a church, and some local industry. The grid consists of a secondary 10kV:400V substation with 98 connected customers with varying load characteristics. The information about grid structure and cable types was obtained through automatic processing from the Geographic Information System at the DSO \cite{shahid21} and an automated mapping of cable types to cable parameters was realized based on available cable data sheets. Using interfaces to the automatically derived digital twin of the low-voltage grid \cite{N2DG_D5.3}, the $C$ matrix was derived, which contains the system equation sets in Equation (\ref{eq:constraint}). 

Figure \ref{fig:tmetopology} shows the grid topology as a graph; since there are no cross-connections between the feeders, a tree topology results. The root of the tree is the low-voltage substation busbar of the secondary transformer. Intermediate nodes represent junction boxes. The leaves of the tree are the Customer Connection Boxes (CCBs), which represent the handover points from the grid operator to the customer and are measured by Smart Meters in the given grid. Although the Smart Meters measure voltages and currents per-phase in this 3-phase grid, we here want to avoid to deal with the issue of wrong phase assignment and therefore we use a single-phase representation of this 3-phase grid and we use the appropriate formulas to do a phase aggregation of the measurements.

\begin{figure}
\centering
  \includegraphics[width=\linewidth]{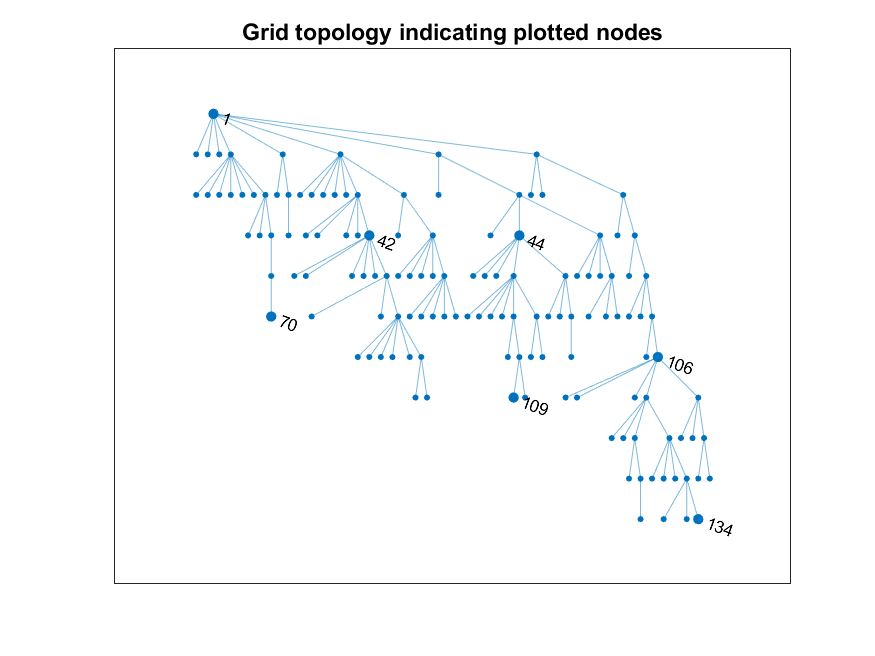}
  \caption{The low-voltage grid topology used in the case-study derived from  a real-life Danish 400V grid. The numbered nodes are used later for plots.}
  \label{fig:tmetopology}
\end{figure}

\subsection{Measurement locations}
The base case for this paper is the actual real-life measurement scenario in which meters measure the voltages at the CCBs as well as the currents into these CCBs. In the PMU model (used for validation of the theory and as a comparison case) the meters measure the actual voltage and current phasors. In reality, the meters follow the EM model, see Section \ref{sect:model}, i.e. they measure voltage magnitude, current magnitude and local phase angles between currents and voltages. The intermediate nodes in the grid topology, Junction Boxes (JB), are not measured in the measurement scenario (and neither in the real-life grid). The LV side of the secondary substation is measured in the real-life grid, but these measurements are only used for the reference grid model.

The smart meter technology deployed in the case study grid allows to collect data in cycles of 6 hrs, with a time resolution of 15 minute intervals. The magnitudes of the voltage and currents are measured, as well as the local phase angle $\phi$. We assume for simplicity that all measurements have successfully been aligned in time, while typically some clock inaccuracies up to 9 seconds are allowed, \cite{IEEEstdC37237}, which would lead to additional errors in the measurand, see \cite{Imad21,ImadACM}.

\begin{table}
\centering
\caption{List of true voltage and current phasors at the specific points in the grid. These will be used later for assessment.}\label{tab:truePhasors}
\begin{tabular}{l|l|l} \hline
Entity& Voltage Phasor (V)& Current Phasor (A)\\ \hline
Subst.&406.60+0.00i & 77.16-11.30i\\
SM-70&406.04-0.08i & 1.68-0.62i\\
SM-109&399.82+1.65i & 1.42-0.10i\\
SM-134&400.00-2.17i & 0.06+0.08i\\
JB-42&405.98-0.27i & 2.72+0.34i\\
JB-44&401.06+1.40i & 21.65-11.79i\\
JB-106&400.88-1.90i & 22.67+1.12i\\
 \end{tabular}
\end{table}

These assumed true values of voltage and current phasors, shown for chosen example locations in the topology in Table \ref{tab:truePhasors}, are later used as comparison base for the estimator and they are also input to the measurement models. The relative measurement errors, $\rho_U$ and $\rho_I$, for voltages and currents are derived from the measurement device type (e.g. measurement device class) as described in the previous section. The standard deviations of voltage magnitude error and current magnitude error are then calculated according to Equation \eqref{eq:prsd}. 
These values as well as a reduced set of true voltages and currents are used as input for the Data Generators, which create a noisy variant of the 'true' voltages and current phasors from the RGM according to the selected meter model (PMU or EM), see Section \ref{sect:assessment}.

\section{Case Study Results}\label{sec:casestudy}
We now apply the assessment approach to the grid and measurement scenario as described in the previous sections. Purposes of the case study are: (1) to validate the calculations and their implementation; (2) to investigate to what extent the approximations via pseudo-measurements (Sect. \ref{sec.modmeas}) and by the 'banana-shape' approximation (Sect. \ref{sect:covar}) for the EM model affect the accuracy of the confidence regions of the estimator; (3) to gain insights into the behaviour of confidence regions for the model estimation results for a practically relevant example case.  

Purpose (1) is achieved via applying the assessment approach using the PMU data generator to derive a so-called hit-rate metric, which is introduced in the next subsection. Reason to use the PMU Data Generator is that it exactly fulfils the theoretical assumptions of the estimator. 
Purpose (2) is addressed by an assessment of the hit-rate metric for the EM Data Generator, and its comparison to the PMU results. Purpose (3) is addressed for the EM Data Generator as practical measurement deployments in the use-case and in almost all of today's distribution grids follow the EM model.

Two metrics are used for the validation and assessment: 1) the average hit-rate, where the hit-rate is defined by the fraction of times that the estimated confidence ellipse from the estimator includes the true value, averaged over all grid nodes (for voltages) respectively grid lines (for currents):
\[HR_k:= \frac{1}{R}\sum_{r=1}^R I(x_{k,r} \subseteq ce_{k,r}) \]
The hit-rate for node $k$, $HR_k$ is averaged over $R$ repetitions, which we then average for all $K$ measurement nodes to a system average hit-rate $HR$.

2) In order to quantify the convergence, we also look at deviation of the hit-rate estimator as follows: 
$$Dev.HR_k=UB\{HR_{k}\}-LB\{HR_{k}\}$$
where $Dev.HR_k$ expresses the deviation between the 95\% upper (UB) and lower (LB) bounds of the hit-rate estimate for node $k$ over $R$ repetitions. Then, for a simpler handle of the metric, we average over all nodes to get the average hit rate deviation:

$$Dev.HR:=\frac{1}{K}\Sigma_k^K Dev.HR_k.$$

\subsection{Validation - Model impact on hit-rates for the PMU Data Generator}
\label{sect:validation}

The introduced grid estimation approach provides not only the estimates of voltage and current phasors but also confidence ellipses for the obtained phasors.  In order to validate the calculation of the confidence ellipses, the  grid scenario from the use-case is fixed and ground truth is obtained,  and a large number of $R=50000$ repetitions of sets of erroneous preprocessed observations are generated by the PMU data generator, and the estimation approach is calculated for each of these $R$ preprocessed observation vectors. 

For the PMU Data generator, the hit-rate (HR) should stochastically converge to the value of $\alpha=95\%$, with $\alpha$ chosen arbitrarily to reflect a 95\% confidence in the estimate, when the error model follows a 2D-Gaussian model.


The left half of the result in Table \ref{tab:hitratesfortwometermodels} confirms the convergence of the hit-rates to $\alpha=95\%$ for the PMU Data Generator. The first row use a baseline assumption on the standard deviations of errors on voltages, currents, and local phase angle, while subsequent rows increase or decrease one of the error standard deviations by a factor of 10, keeping the other two on their baseline value. These standard deviations are used in the Data Generator and the derived covariance matrices $\Sigma_1$ and $\Sigma_2$ are channeled as input into the estimator; therefore the estimator is made aware of the change of the standard deviations in the Data Generators.


\begin{table*}
\centering
\caption{Hit-rate assessment using N=50000 repetitions: results are shown for both Data Generators, PMU in the middle and EM in the right; the different rows show different variations of the base case (first numerical row) of standard deviations of the error added to the measurements.}\label{tab:hitratesfortwometermodels}
\rotatebox{-90}{
\begin{tabular}{|c|cc|cc|cc|cc|cc|}
 \hline
 \multicolumn{3}{|c|}{Measurement Error STDs}&\multicolumn{4}{|c|}{PMU Data Generator}&\multicolumn{4}{|c|}{EM Data Generator}\\ \hline
 \multicolumn{1}{|c|}{Voltage}&\multicolumn{2}{|c|}{Currents} & \multicolumn{2}{|c|}{Voltage}&\multicolumn{2}{|c|}{Currents}&\multicolumn{2}{|c|}{Voltage}&\multicolumn{2}{|c|}{Currents}\\ \hline
  $\sigma^2_{u,n}$& $\sigma^2_{i,n}$&$\sigma^2_{\phi,n}$& avg HR & Dev.HR. & avg. HR.  & Dev.HR.& avg. HR. & Dev.HR. & avg. HR. & Dev.HR.\\ 
  
  [$V^2$]& [$A^2$]&[$rad^2$]& [\%] & [\%] & [\%] & [\%]& [\%] & [\%] & [\%] & [\%]\\ \hline
 2.41 & 2.1$\cdot 10^{-3}$ & 1$\cdot 10^{-4}$ & 94.92 & 0.39 & 94.99 & 0.38 & 94.00 & 0.42 & 95.36 & 0.36 \\ \hline
 10x & . & .  & 95,02 & 0,38 & 95,00 & 0,38 & 87,54 & 0,58 & 95,33 & 0,37\\ 
 . & 10x & .  & 95,12 & 0,38 & 94,99 & 0,38 & 93,86 & 0,38 & 94,45 & 0,40\\ 
 . & . & 10x  & 95,09 & 0,38 & 94,99 & 0,38 & 74,70 & 0,43 & 92,05 & 0,41\\ \hline
 0.1x & . & . & 95,04 & 0,38 & 94,99 & 0,38 & 99,89 & 0,41 & 95,35 & 0,37\\
 . & 0.1x & . & 95,05 & 0,38 & 95,00 & 0,38 & 94,07 & 0,41 & 96,02 & 0,34\\
 . & . & 0.1x & 95,04 & 0,38 & 95,00 & 0,38 & 93,77 & 0,41 & 92,33 & 0,40
 \\ 
 \hline
\end{tabular}}
\end{table*}
\normalsize

\subsection{Assessment of hit-rate when using the EM Data Generator}
When the assessment uses the EM data generator, there are two approximations that deviate from the theoretical and accurate case: (1) the Pseudo-measurement for the voltage angle, which is taken in the data preparation in the EM data generator (Sect. \ref{sec.modmeas}); (2) the assumption in the estimator of Gaussian noise, which is only approximately true, see Sect. \ref{sect:covar}. These two approximations will influence the hit-rates of the confidence regions.

The right half of Table \ref{tab:hitratesfortwometermodels} shows the impact of these two approximations, which depends on parameters of the measurement errors. In addition the grid scenario will influence the impact of the pseudo-measurements, so here the results have to be interpreted in the context of the realistic distribution grid of the use-case. 

The base case of measurement error standard devations, which is derived from realistic measurement device parameters, is specified in the top numerical row. For this base case, the hit-rates for voltage and current confidence regions are very close to the theoretical value of $95\%$, the voltage hit rate is $1\%$ lower while the current hit-rate is slightly higher than $95\%$. When decreasing the standard deviations of the voltage or of the current, see third and second row from the bottom, respectively, the hit-rates stay at the base level approximately or even increase slightly. When decreasing the standard deviation of the local phase angle error between voltage and currents, both hit-rates decrease slightly - the latter is suspected to be an impact of the pseudo-measurements for voltage angles, which are suspected to take over more strongly when the standard deviation of the local phase angle error decreases.

When increasing the standard deviations by a factor of 10 individually (numerical rows 2 to 5 in table), hit-rates for voltages and currents both decrease, however to a different extent. Note that a factor of 10 on the standard deviations corresponds approximately to measurement errors of 
$10\%-30\%$ (or even much more for current magnitudes, as the actual smaller values of the current magnitudes further out in the grid are not calculated in here when using a constant standard deviation), which is typically beyond practical relevant measurement deployments, so this is an extreme case to demonstrate the impact.

When increasing the standard deviation of the voltage magnitude error by a factor of 10, hit-rates for current estimator condidence regions are only marginally affected, while the hit-rates for the voltage estimator drop by almost $7\%$. When increasing the standard deviation of the current magnitude errors, we see the inverse situation: Voltage estimator hit-rates are only marginally reduced while the impact on current estimator hit-rate becomes noticable, while still being less than $1\%$ reduced.
The increase of the standard deviation of the error on the local phase angle has the strongest impact: hit-rate of voltage phasors drop by almost $20\%$ and hit-rates of current-phasors drop by a bit more than $3\%$ compared to the base-case.

In summary, the realistic base case leads to hit-rates very close to the theoretical $95\%$ while extreme scenarios of huge measurement errors may reduce these hit-rates strongly. The practical applicability of the novel approach to EM measurements in realistic settings is confirmed by the analysis.

\subsection{Voltage and current estimates under the EM model}
In the subsequent analysis the EM data generator is used with the relative errors according to the base-case parameters (so from now on, standard deviations of errors on the current magnitude are in fact proportional to the true current magnitude of the currents at that measurement location, in contrast to the analysis in the previous section).

Figures \ref{fig:volt_est_meters} to \ref{fig:current_est_cb} show the estimated voltage or current phasors (visualized by a square) in one repetition in comparison to the true value (diamond) and the estimated confidence ellipse. The first figures show the estimates at customer connection boxes, which also provide measurements (visualized by plus symbol) as input to the estimation approach. For the example results, we chose meters with IDs $\{70, 109, 134\}$, see Figure \ref{fig:tmetopology} for grid location. 
The measurements of voltages in Figure \ref{fig:volt_est_meters} are actually all lying on the x-axis, since the EM metering model does not allow to measure absolute phase angles - instead the phase angle $\theta$ of the voltage phasor is assumed to be 0 in the measurement. Figure \ref{fig:current_est_meters} shows the current estimates with confidence ellipses for the same meters.

\begin{figure*}
\minipage{0.32\textwidth}
  \includegraphics[width=1\columnwidth]{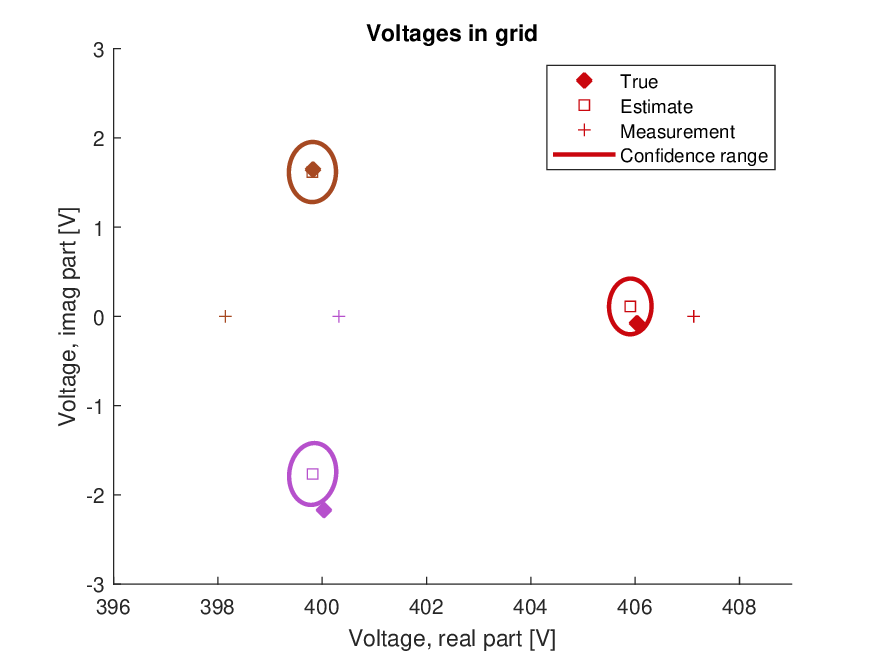}
  \caption{Voltage estimates, measurements, true values and confidence regions for selected customer connection boxes (providing measurements).}
  \label{fig:volt_est_meters}
\endminipage\hfill
\minipage{0.32\textwidth}
  \includegraphics[width=1\columnwidth]{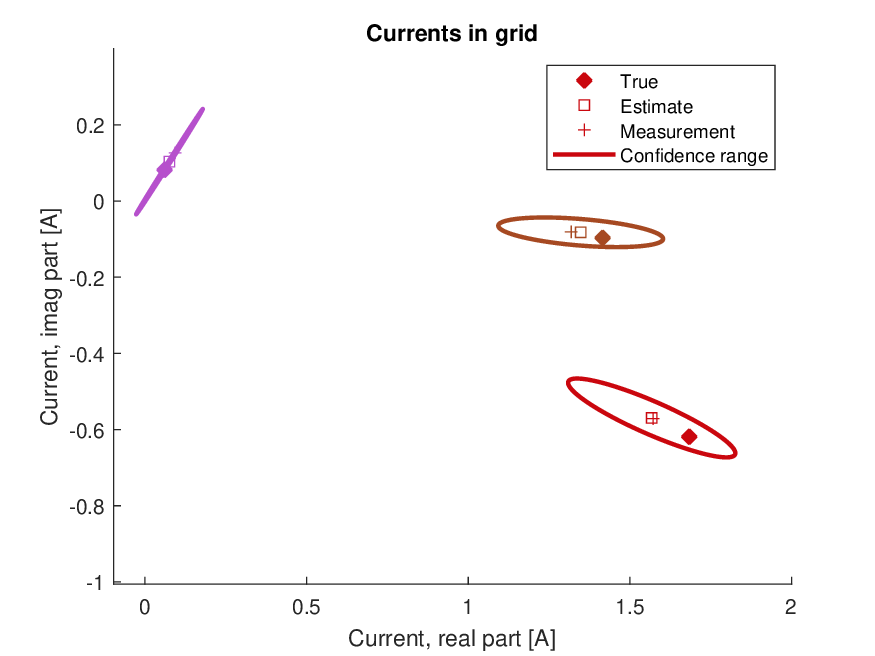}
  \caption{Current estimates, measurements, true values and confidence regions for selected customer connections (providing measurements).}
  \label{fig:current_est_meters}
\endminipage\hfill
\minipage{0.32\textwidth}%
  \includegraphics[width=1\columnwidth]{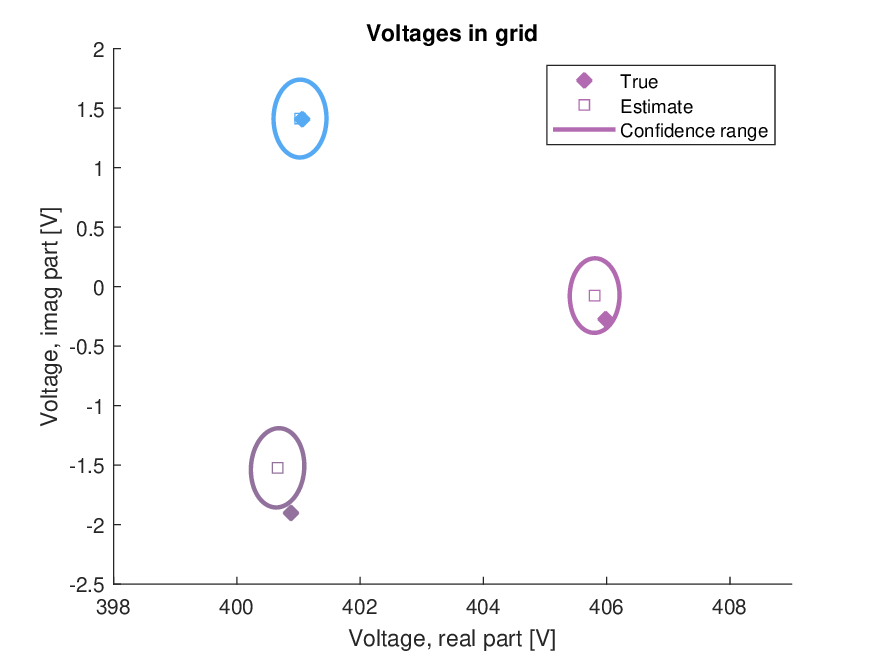}
  \caption{Voltage estimates, true values and confidence regions for selected junction boxes where measurements are not available.}
  \label{fig:volt_est_cb}
\endminipage
\end{figure*}



Figures \ref{fig:volt_est_cb} and \ref{fig:current_est_cb} show a similar result for intermediate junction boxes, with ID's $\{42,44,106\}$, which do not provide any measurements to the estimation approach.


\begin{figure}
    \centering
  \includegraphics[width=0.8\columnwidth]{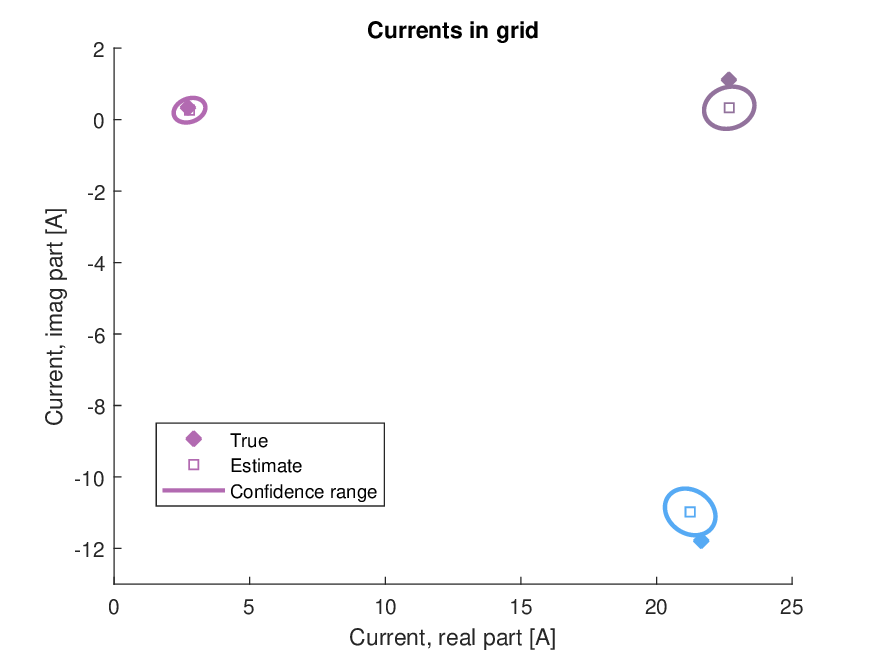}
  \caption{Current estimates, true values and confidence regions at lines into selected junction boxes where measurements are not available. 
  }
  \label{fig:current_est_cb}
\end{figure}




\subsection{Analysis of noise on confidence regions}
In grid monitoring we are ultimately interested in estimation of voltage and current magnitudes, and possibly on phase angles. In order to assess how these estimation results are impacted by different measurement error magnitudes, we define the $95\%$ confidence range for voltage magnitudes as the difference of the largest and smallest voltage magnitude on the confidence ellipse $\Delta C$, indexed either for voltage (U) or current (I): 
$$\Delta C_U =\max_{ConfEllipse} |U| - \min_{ConfEllipse} |U|.$$
We assess in the following, how these confidence ranges are impacted by changing error magnitudes.

\subsubsection{Change of error standard deviation on measurements of voltage magnitude}

We first change the standard deviation, $\sigma_U$ of the normally distributed error on the voltage magnitude. Figure \ref{fig:varVarVMagOnVoltage} shows that the confidence range for the estimated voltage magnitude for all 7 considered grid locations changes almost linearly with $\sigma_u$ within the considered range of the voltage measurement error standard deviation between 1V and 4V. Furthermore, the difference of confidence ranges between different grid locations is small - and reduced further for larger $\sigma_U$. Although being measured by Smart Meters as input to the estimator, Customer Connection Boxes (dashed) show a slightly larger confidence range for the voltage as compared to junction boxes (non-measured, dotted) and the substation itself (also not measured for the estimation, solid).

\begin{figure}
\centering
  \includegraphics[width=0.8\columnwidth]{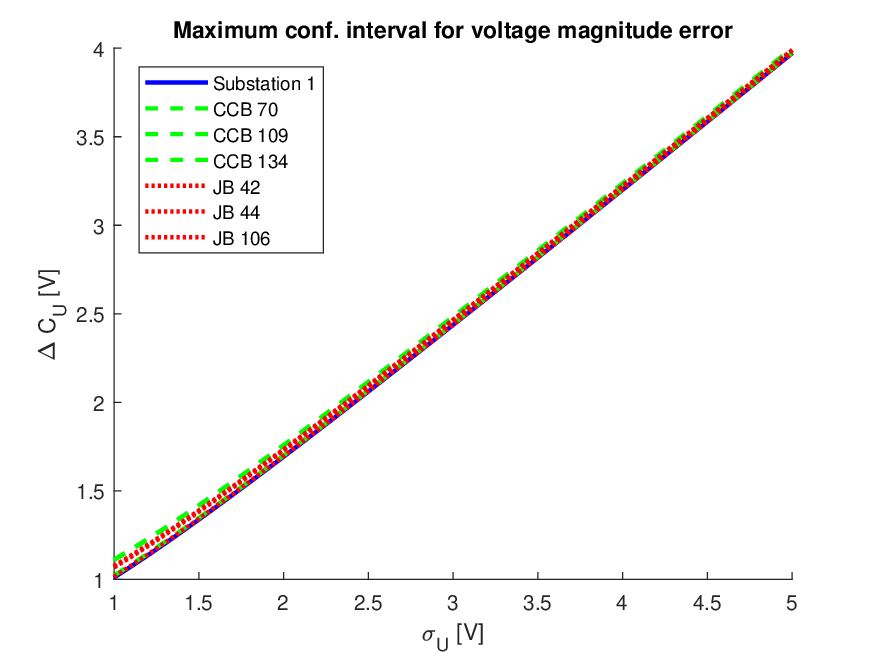}
  \caption{Impact of voltage estimation on selected nodes for increasing noise on voltage magnitude. }
  \label{fig:varVarVMagOnVoltage}
\end{figure}

A very different behavior is obtained for the confidence range of the current magnitudes, shown in Figure \ref{fig:varVarVMagOnCurrent}: Varying the standard deviation of the voltage measurements in the range between 1V and 4V only has negligible impact on the confidence ranges of the estimated currents, the shown curves are almost horizontal lines. In contrast to the voltage magnitudes, the confidence ranges of the current magnitudes however depends strongly on the grid location (here the line feeding into the substation bus bar, JB or CCB).
Furthermore, the different CCBs show different confidence ranges for the current magnitudes, partly resulting from the quite different true current phasors, see Table \ref{tab:truePhasors}.

\begin{figure}
\centering
  \includegraphics[width=0.8\columnwidth]{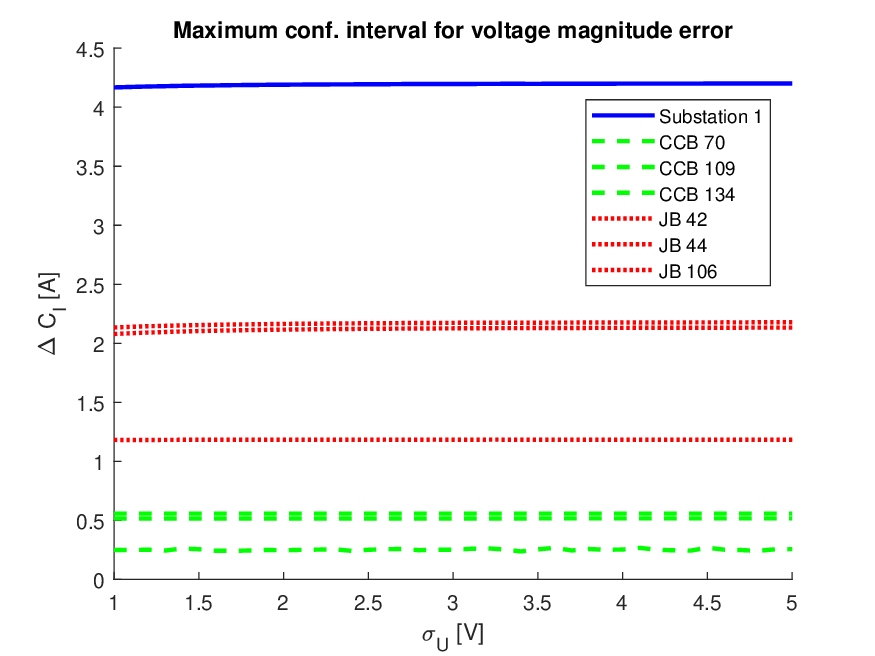}
  \caption{Impact of current estimation on selected nodes for increasing noise on voltage magnitude.}
  \label{fig:varVarVMagOnCurrent}
\end{figure}







\subsubsection{Change of error standard deviation on measurement of current magnitude}

For this set of experiments, we move away from a fixed relative error on the current magnitude measurements across the measured CCBs, and instead set a fixed standard deviation, $\sigma_I$ of the error on current magnitude across all measured grid lines. Figure  \ref{fig:varVarIMagOnVoltage} shows the impact on the confidence range of the voltage magnitude estimates, while Figure \ref{fig:varVarIMagOnCurrent} shows the impact on the confidence range of the estimated current magnitudes. Compared to varying the error standard deviation on voltage magnitudes in the previous subsection, the impact is more widely varying depending on inspected grid location. 

\begin{figure}
    \centering
  \includegraphics[width=0.8\columnwidth]{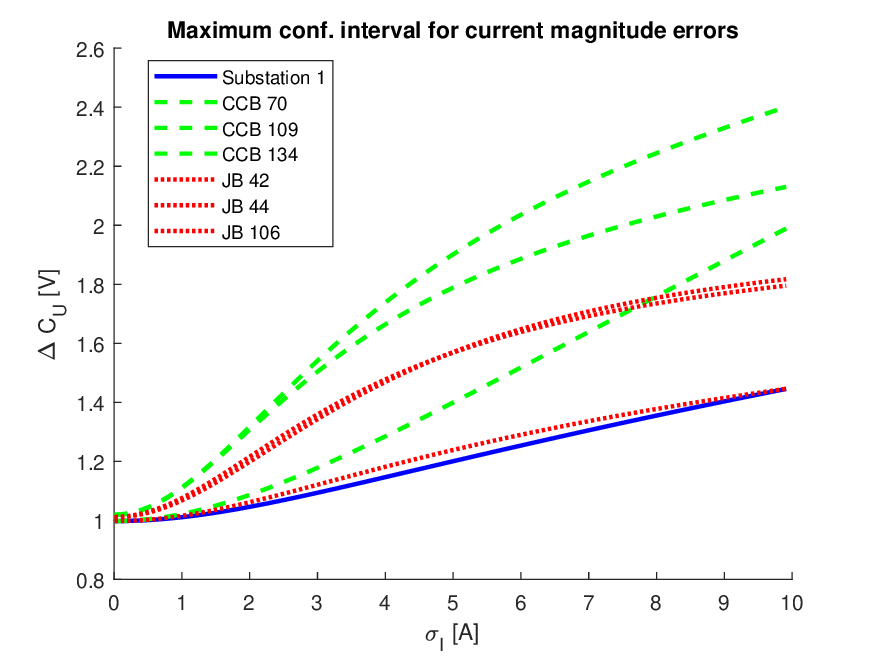}
  \caption{Impact of voltage estimation on selected nodes for increasing noise on current magnitude measurements.}
  \label{fig:varVarIMagOnVoltage}
\end{figure}

\begin{figure}
    \centering
  \includegraphics[width=0.8\columnwidth]{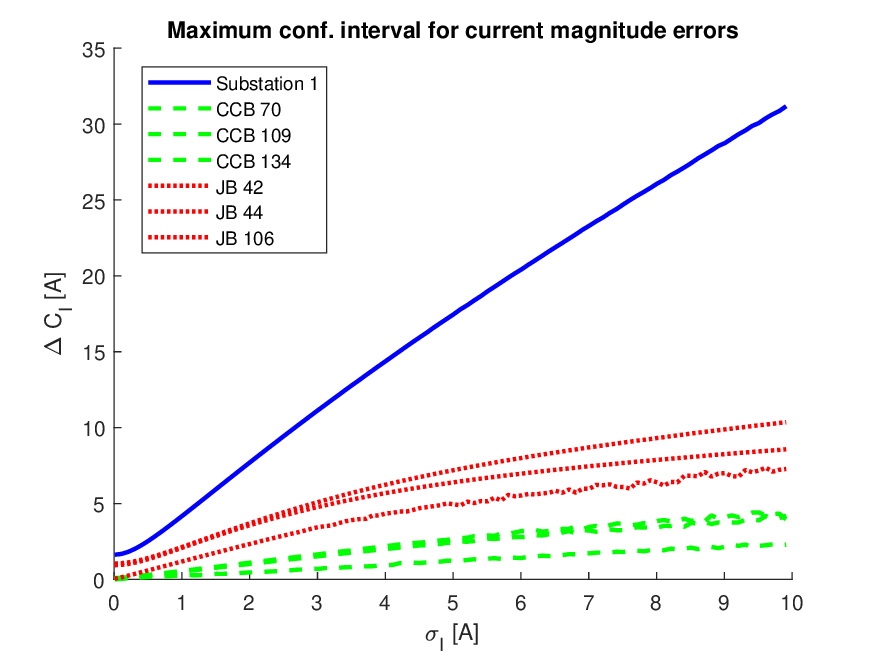}
  \caption{Impact of current estimation on selected nodes for increasing noise on current magnitude measurements.}
  \label{fig:varVarIMagOnCurrent}
\end{figure}

\subsubsection{Change of error standard deviation on measurement of local phase angles}
When varying the error magnitude on the local phase angle measurement (here a range of $\sigma_\phi$ up to approximately double of its value in the base scenario is applied), results are not visibly impacted for the parameter ranges. Within the specific range of $\sigma_\phi\in[0,0.018]$ rad confidence ranges for voltage and current estimates do not change within four digits accuracy.


%
%

We also investigated the confidence range of voltage and current magnitudes when visualizing over the grid topology graph. However, as the confidence ranges for voltage magnitudes show only little differences across different grid locations, and the confidence ranges for current magnitudes are largely influenced by the true current magnitude, the results are not interesting to be shown here. However, this will change when changing the measurement scenario, e.g. assuming some CCBs as not being measured, or adding the measurement at the substation. The latter studies have been started and will be published in subsequent work.

\section{Conclusion and Outlook}
 \label{sect:conclusions}

In this paper we presented a methodology to derive and assess confidence regions of estimated voltage and current phasors in an electrical distribution grid. Typical measurement devices in low and medium voltage grids are not clock synchronized to high precision level and therefore only allow to measure local phase angles. Therefore, the paper introduced the EM metering model, which is subsequently compared to a PMU metering model. While the PMU metering model fulfills the theoretical assumptions for the estimator, two approximations need to be introduced for the EM metering model: (1) constant zero-valued pseudo measurements on the phase angle of voltages; (2) approximations of the measurement error by 2-D Gaussian noise. 
Furthermore, not all desirable grid locations may be observed by measurement devices in distribution grids. 

In a realistic use-case of a single-phase representation of a real distribution grid, we showed that the estimation and related confidence ellipses show the theoretically expected stochastic convergence for the PMU model, while there are some deviation from the specified confidence level for a more realistic EM measurement model. However, these deviations are only significant in case of extremely large measurement errors.  A confidence range was used to asses the impact of increasing measurement errors on measured voltage magnitudes, current magnitudes and phase angle measurements, respectively, and example results from the realistic use-case have been reported.

Further reduced sets of measurement locations, e.g. in which some customers do not have smart meters, or some smart meters only measure voltages (which are less privacy sensitive), and also additional measurement device locations at the secondary substation will be investigated in future case studies. In this context, the generation of extended types of pseudo-measurements and the derivation of properties of the measurement error of such pseudo-measurements will also need to be studied. In addition, a characterization and inclusion of errors in the structural information and cable properties of the grid will be interesting. Finally, the extension of the use-case study to use a 3-phase plus neutral model will be interesting; the presented mathematics hold as the grid equations stay linear \cite{florinpaper}, while additional complexity results from the wrong-phase assignment cases of measurements of voltages and currents.

Recent technological developments with 5G and beyond, and the installment of these communication interfaces in modern smart meters, enables a higher accurate clock synchronization among the meters. For example \cite{Nokia_clocksynch5G_21} discusses the requirements and use of Precision Time Protocols (PTP) to achieve clock synchronization as low as $\pm1.5\mu$sec. for smart grid applications, and \cite{Patel21} concludes in 5G end-to-end systems clock synchronization can within a 5G-TSN network be kept within boundaries of 900ns. 5G enabled smart meters are, however, at this stage not widely deployed, but the industrial trend moves towards using such networks in next generation smart meters, hereby enabling applications with high requirement for clock synchronization such as DSSE. The advancement of technology may lead to the partial availability of voltage phasor measurements. Use-cases of heterogeneous measurement types, in the terminology of this paper, some measurement locations with PMU meters, others with EM meters, will therefore be interesting for future studies.

 \bibliographystyle{elsarticle-num} 
  \bibliography{ref.bib}

\appendix
\section{}
\begin{proposition}
\label{propMeanCovariacePseudocovariance}
Let $Z=(z+\epsilon_z) e^{j(\nu+\epsilon_\nu)}$, $z,\nu\in\mathbb{R}$, and $\epsilon_z$ and $\epsilon_\nu$ be independent random variables with $\mathbb{E}[\epsilon_z] = \mathbb{E}[\epsilon_\nu] = 0$, $\mathrm{Var}(\epsilon_z) = \sigma_z^2$ and $\mathrm{Var}(\epsilon_\nu) = \sigma_\nu^2$. Then the mean, variance and pseudo-variance of $Z$ are given by
\begin{equation*}
    \mathbb{E}[Z] =  z e^{j\nu} c_{\epsilon_\nu}(1),  \quad    \mathrm{Var}(Z) = (1-|c_{\epsilon_\nu}(1)^2|) z^2 + \sigma_z^2
\end{equation*}
and 
\begin{equation*}
    \mathrm{PVar}(Z) = e^{2j\nu}((z^2+\sigma_z^2)c_{\epsilon_\nu}(2) - z^2 c_{\epsilon_\nu}(1)^2),
\end{equation*}
where $c_{\epsilon_\nu}(t) = \mathbb{E}[e^{jt\epsilon_\nu}]$ is the characteristic function of $\epsilon_\nu$. 

\end{proposition}

\begin{proof}
Using that $\epsilon_z$ and $\epsilon_\nu$ are independent and $\mathbb{E}[\epsilon_z]=0$, we get that
\begin{equation*}
    \mathbb{E}[Z] = \mathbb{E}[z+\epsilon_z]\mathbb{E}[e^{j(\nu+\epsilon_\nu)}] = z e^{j\nu} c_{\epsilon_\nu}(1),
\end{equation*}
\begin{equation*}
    \mathbb{E}[|Z|^2] = \mathbb{E}[(z+\epsilon_z)^2] = z^2 + \sigma_z^2
\end{equation*}
and
\begin{equation*}
    \mathbb{E}[Z^2] = \mathbb{E}[(z+\epsilon_z)^2]\mathbb{E}[e^{2j(\nu+\epsilon_\nu)}] = (z^2+\sigma_z^2) e^{2j\nu} c_{\epsilon_\nu}(2).
\end{equation*}
The formulas for the variance and pseudo-variance in the general case follows by inserting the expected values into $\mathrm{Var}(Z) = \mathbb{E}[|Z|^2] - |\mathbb{E}[Z]|^2$ and $\mathrm{Var}(Z) = \mathbb{E}[Z^2] - \mathbb{E}[Z]^2$.  
\end{proof}
\begin{corollary} Let $Z$ defined as in Proposition~\ref{propMeanCovariacePseudocovariance} and let $\epsilon_\nu$ be normally distributed. Then
\begin{equation*}
  \mathbb{E}[Z] =  z e^{j\nu} e^{-\sigma_\nu^2/2},  \quad      \mathrm{Var}(Z) = (1-e^{-\sigma_\nu^2}) z^2 + \sigma_z^2
\end{equation*}
and 
\begin{equation*}
    \mathrm{PVar}(Z) = e^{2j\nu}((z^2+\sigma_z^2)e^{-2\sigma_\nu^2} - z^2 e^{-\sigma_\nu^2}).
\end{equation*}
\end{corollary}
\begin{proof}
Use the characteristic function for the normal distribution in Proposition~\ref{propMeanCovariacePseudocovariance}.
\end{proof}


\end{document}